\documentclass[letterpaper, 10 pt, conference]{ieeeconf}  

\usepackage{graphicx,keyval,trig,url} 
\usepackage[T1]{fontenc}     
\usepackage[utf8]{inputenc}	
\usepackage{graphicx}

\usepackage{amsthm}
\usepackage{amssymb}
\usepackage{amsmath}
\usepackage{mathtools}
\usepackage[mathscr]{eucal}
\usepackage{lipsum}
\usepackage{cuted} 
\usepackage{blindtext}
\usepackage{color}
\usepackage{dsfont}

\usepackage{epstopdf}
\usepackage{bm}   

\usepackage{color}
\usepackage{xcolor}

\usepackage{multirow}
\usepackage{longtable}

\usepackage{mathtools}
\usepackage[colorinlistoftodos]{todonotes}

\newtheorem{theorem}{Theorem}
\newtheorem{prop}{Proposition}
\newtheorem{assumption}{Assumption}

\newtheorem{problem}{Problem}

\IEEEoverridecommandlockouts                              

\title{\LARGE \bf
The Proportional Integral Notch and Coleman Blade Effective Wind Speed Estimators and Their Similarities
}

\author{Yichao Liu$^{1}$, Atindriyo K. Pamososuryo$^{1}$, Sebastiaan P. Mulders$^{1}$,\\ Riccardo M.G. Ferrari$^{1}$ and Jan-Willem van Wingerden$^{1}$
\thanks{This research was conducted in cooperation with Vestas Wind Systems A/S. It was also partially supported by the European Union via a Marie Sklodowska-Curie Action (Project EDOWE, grant 835901).
$^{1}$Delft University of Technology, Delft Center for Systems and Control, Mekelweg 2, 2628 CD Delft, The Netherlands.
        {\tt\small \{Y.Liu-17, A.K.Pamososuryo, S.P.Mulders, R.Ferrari, J.W.vanWingerden\}@tudelft.nl}.}}

\begin{document}

\maketitle
\thispagestyle{empty}
\pagestyle{empty}

\begin{abstract}
\noindent The estimation of the rotor effective wind speed is used in modern wind turbines to provide advanced power and load control capabilities. However, with the ever increasing rotor sizes, the wind field over the rotor surface shows a higher degree of spatial variation. A single effective wind speed estimation therefore limits the attainable levels of load mitigation, and the estimation of the Blade Effective Wind Speed (BEWS) might present opportunities for improved load control.
This letter introduces two novel BEWS estimator approaches: A Proportional Integral Notch (PIN) estimator based on individual blade load measurements, and a Coleman estimator targeting the estimation in the non-rotating frame. Given the seeming disparities between these two estimators, the objective of this letter is to analyze the similarities between the approaches. 
It is shown that the PIN estimator, which is equivalent to the diagonal form of the Coleman estimator, is a simple but effective method to estimate the BEWS. 
The Coleman estimator, which takes the coupling effects between individual blades into account, shows a more well behaved transient response than the PIN estimator.

\end{abstract}

\section{Introduction}
\noindent Over the past decade, wind energy has grown exponentially in the global energy mix, benefiting from scientific advancement, reduced costs and governments' subsidy schemes. 
The Global Wind Energy Council~\cite{GWEC_2021} reported that more than 90\,GW of new wind power was deployed in 2020, which exhibits a global growth of 53\,\% compared to 2019. 
This growth is partially driven by the fact that the physical dimensions of wind turbines dramatically increased, which resulted in a rising demand for optimization of wind turbine control systems.

For larger turbines, the wind inflow conditions over the rotor area demonstrate a high degree of spatial variability. As a result, the need for more accurate and granular wind speed information is becoming ever more prominent for the design of effective control algorithms. The wind speed measurement from the conventional anemometer is generally omitted in the control system, as it is a single point-wise measurement disturbed by rotor induction~\cite{Soltani_2013}. 
Therefore, a wind sensing technique that leverages the wind turbine rotor as a generalized anemometer, to provide an estimate of the effective wind inflow conditions, is considered as a viable solution to address this challenge.

Wind sensing techniques are based on the idea that wind state changes are reflected by prior turbine knowledge, real-time measurements and control signals. 
For instance, the widely-used Immersion and Invariance (I\&I) estimator~\cite{Ortega_2013} exploits prior knowledge on the (aerodynamic) turbine properties, combined with an angular speed measurement and the applied generator torque control signal to estimate the Rotor Effective Wind Speed (REWS). 
Liu \textit{et al}.~\cite{Liu_LCSS2021} revisited the I\&I estimator, added an integral term, and proved that the technique shows a high degree of similarity to the popular torque balance estimation method~\cite{Soltani_2013,Bossanyi_2000,Ostergaard2007}.


The previously mentioned methods still return a single point-wise estimate of the effective wind speed across the rotor disc. The increasingly common usage of load sensors on wind turbines however, forms an opportunity for more advanced wind speed estimation solutions.  Bottasso \textit{et al}.~\cite{Bottasso_2018} developed a load-sensing approach to estimate the REWS, where several wind characteristics are estimated (\emph{e.g.,} wind shear and yaw misalignment). Liu \textit{et al}.~\cite{Liu2021} proposed a Subspace Predictive Repetitive Estimator (SPRE) to identify the periodic wind flow on an individual blade, while also providing approximations on wind shear and wake impingement based on blade load measurements. 
Although the algorithm has proven to be very effective, the SPRE approach introduces significant delays in the blade effective wind speed (BEWS) estimation, thus making it less suitable for closed-loop control.

In this letter, two novel load-sensing approaches to estimate the periodic wind flow on an individual blade are proposed. 
The first method is called Proportional Integral Notch (PIN) estimator, in which an azimuth-dependent cone coefficient is defined to reflect the one-to-one relation between the wind speeds and the blade loads. It is also assumed that the periodic wind speed on an individual blade is a superposition of REWS and zero-mean BEWS acting on the rotating frequencies. The proportional-integral component of the estimator is used to estimate the REWS component at 0P, while the gain-scheduled notch is added to identify the zero-mean BEWS component at higher frequencies.
The second approach employs the Coleman transformation to translate blade load signals from the rotating frame into the fixed frame. This transformation is often used in individual pitch control (IPC) implementations for blade fatigue load reductions~\cite{Bossanyi2003,Mulders_2019}. The Coleman estimation technique shows resemblance with the PIN estimator, making it compelling to understand the similarities of both schemes in terms of their structure and performance. To this end, as a core contribution of this letter, a frequency-domain analysis has been performed to understand the similarities under varying wind conditions.

This study provides the following contributions: 
\begin{enumerate}
\item Proposing a novel PIN estimator to estimate the BEWS for wind turbines.
\item Accounting for the coupled blade dynamics by developing a new Coleman estimator. 
\item Showcasing the similarities between the PIN and the Coleman estimators via theoretical analyses and numerical simulations.
\end{enumerate}
The remainder of this letter is organized as follows:
Section~\ref{sec:2} introduces the wind turbine model considered in this research and the azimuth-dependent cone coefficient used by the proposed wind speed estimators. Section~\ref{sec:3} theoretically formalizes the estimation techniques and demonstrates the similarities of both schemes.
In Section \ref{sec:4}, high-fidelity wind turbine simulations are carried out to demonstrate the estimators' performance and to verify the theoretical results. Finally, conclusions are drawn in Section~\ref{sec:5}.


\section{Wind turbine model and its cone coefficient}\label{sec:2}
\noindent This section introduces the wind turbine model and the azimuth-dependent \emph{cone coefficient} \cite{Bottasso_2018} which are used to establish the wind speed estimation schemes. 
For a horizontal axis wind turbine, the \emph{azimuth-dependent cone coefficient} is defined as:
\begin{equation}
C_{\mathrm{m},i}(\lambda_{i}, q_{i}, \psi_i) := 
\frac{m_{i}(\lambda_{i}, q_{i}, \psi_i)}{\frac{1}{2} \rho AR U_{i}^2}
\, ,
\label{eq:cone coefficient}
\end{equation}
where $\rho$ is the air density, $A$ the rotor swept area, and $U_i$ represents the BEWS. 
The measured out-of-plane blade root bending moment (MOoP) is denoted by $m_{i}$, which is --~amongst other variables~-- a function of the azimuth angle $\psi_i = \psi + 2\pi(i-1)/3$, where $i=\left\{1,\,2,\,3\right\}$ is the blade index for a three-bladed wind turbine and $\psi$ the azimuth angle of the first blade. 
The cone coefficient is represented by $C_{\mathrm{m},i}$ and depends on the azimuth position of blade $i$.
Here, $C_{\mathrm{m},i}$ is a nonlinear function of the \emph{tip-speed ratio}
\begin{equation}
\lambda_i := \frac{\omega_\mathrm{r}R}{U_i}
\, ,
\label{eq:lambda}
\end{equation}
where $\omega_\mathrm{r}$ and $R$ are the rotor speed and the rotor radius, respectively. 
In Eq.~\eqref{eq:cone coefficient}, $q_i$ represents the \emph{blade effective dynamic pressure}
\begin{equation}
q_i := \frac{1}{2}\rho U_i^2
\, .
\label{eq:q}
\end{equation}
The shape of the $C_{\mathrm{m},i}$ surface is determined by the structural design of the wind turbine, and is typically derived via either high-fidelity numerical simulations or experimental tests.
Note that $C_{\mathrm{m},i}$ is also a function of the blade pitch angle~\cite{Liu2021}. However, without loss of generality, the pitch angle is assumed to be constant throughout this study.

The wind turbine model considered in this letter is the National Renewable Energy Laboratory (NREL)~5\nobreakdash-MW reference wind turbine~\cite{Jonkman_2009}. 
Its $C_{\mathrm{m},1}$ curve is shown in Fig.~\ref{Cm table} for illustration purposes, which is obtained from steady-state wind turbine simulations, in which the turbine is subjected to a sheared wind profile with a constant velocity of $8$\,m/s.

Once the cone coefficient is derived for each blade and over the operating conditions of interests, Eq.~\eqref{eq:cone coefficient} can be used to estimate the MOoP as:
\begin{equation}
\hat{m}_i(\lambda_{i}, q_{i}, \psi_i) = \frac{1}{2}\rho AR \hat{U}_{i}^2C_{\mathrm{m},i}(\lambda_{i}, q_{i}, \psi_i)
\, .
\label{eq:nonlinearity}
\end{equation}

\begin{figure}
\centering 
\includegraphics[width=1\columnwidth]{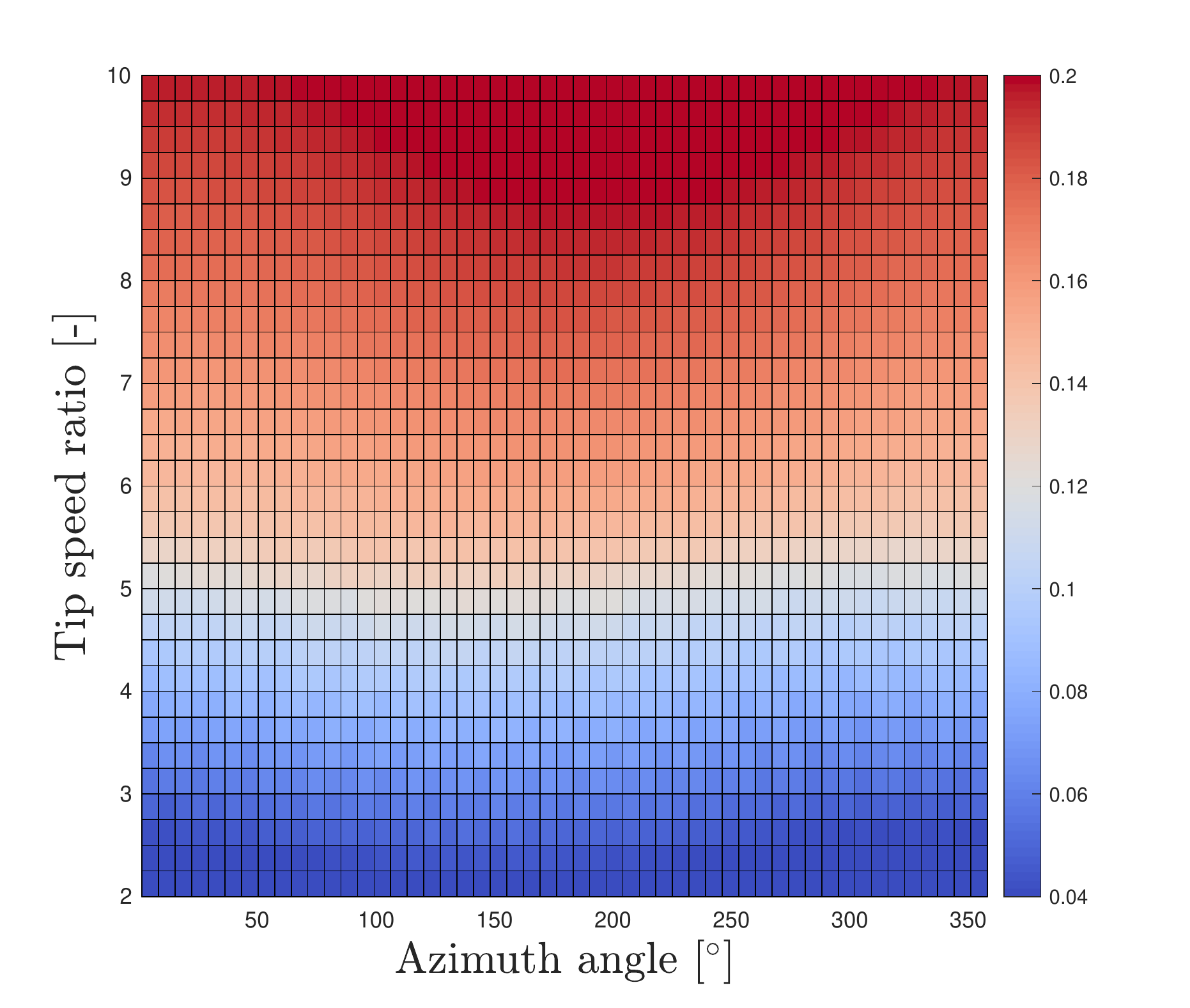}
\caption{Cone coefficient surface $C_{\mathrm{m},1}$ for blade 1 of the NREL~5-MW wind turbine model for a uniform constant wind speed of $8$\,m/s.}
\label{Cm table} %
\end{figure}
Based on the wind turbine model and its azimuth-dependent cone coefficient, the BEWS estimator problem is formalized in Section~\ref{sec:3}.





\section{Blade effective wind speed estimator}\label{sec:3}
\noindent First, the following assumptions are formulated for the wind speed estimation problem.
\begin{assumption}\label{ass:cp_smooth}
$C_{\mathrm{m},i}$ is a \textit{known} and smooth nonlinear function of class $\mathbb{C}^{n}$, i.e., its $0^{\mathrm{th}}$ through $n^{\mathrm{th}}$ derivatives are continuous, with $n\in\mathbb{Z}^{+}$.
\end{assumption}

\begin{assumption} \label{ass:constant_U}
The blade effective wind speed for an individual blade $U_i$ and for the rotor disk $\bar{U}$ are regarded as positive and \textit{unknown} signals.
\end{assumption}

\begin{assumption}\label{ass:measurable_Tg_wr}
The signals $m_i$ and $\omega_\mathrm{r}$ are assumed to be measured. The wind turbine operates at a constant rotor speed, i.e., $\dot{\omega}_\mathrm{r}(t)=0$.
\end{assumption}

The wind speed estimation problem solved in this letter is thus formulated as:

\begin{problem}\label{prob:alter_proof}
Given the nonlinear relation~\eqref{eq:nonlinearity} for estimation of the MOoP, find asymptotically convergent estimates $\hat U_i$ and $\hat{\bar{U}}$ for the wind speed such that:
\begin{equation}
\lim_{t\to\infty} \hat{U}_i(t) = U_i, \hspace{4mm} \lim_{t\to\infty} \hat{\bar{U}}(t) = \bar{U}
\, .
\label{eq:estimate}
\end{equation}
\end{problem}
\noindent The following subsections outline two wind speed estimation techniques. The last subsection concludes by performing a similarity analysis.

\subsection{Proportional integral notch estimator}
\noindent In our previous work~\cite{Liu_LCSS2021}, we revisited the widely-used I\&I estimator and extended the proportional-only structure by including an integral correction term.
Inspired by the extended I\&I estimator which returns an estimate of the REWS, this letter proposes a novel wind speed estimator called the Proportional Integral Notch (PIN) estimator to estimate the BEWS.
The frequency content of the wind speed experienced by each blade is composed of specific harmonics that are related to (integer multiples of) the rotor rotational frequency. This is due to, e.g., the wind shear and tower shadow. The BEWS can thus be estimated by amplifying those specific frequencies via a feedback control structure. 
Therefore, the periodic wind speed over an individual blade $i$ can be estimated as
\begin{equation}
\begin{cases}
\epsilon_i=\hat{m}_i\left(\omega_\mathrm{r},  \hat{U}_{i}, \hat{q}_i, \psi_i \right)-m_i,\\
\hat{U}_i(s) = K(s)\epsilon_i(s) = \left( k_\mathrm{p} K_\mathrm{N}(s) + k_\mathrm{i}/s \right) \epsilon_i,
\end{cases}
\label{PIN equation}
\end{equation}
where $\left\{ k_\mathrm{p},\, k_\mathrm{i}\right\} \in\mathbb{R}^{+}$ are the proportional and integral gains. 
The error between the estimated and measured MOoP is indicated by $\epsilon_i$.
The notch transfer function $K_\mathrm{N}(s)$ is scheduled by the rotor speed $\omega_\mathrm{r}$ and is defined as:
\begin{equation}
\label{eq:Notch-peak}
K_\mathrm{N}(s) = \frac{2 w_\mathrm{r} s}{s^2+w_\mathrm{r}^2}.
\end{equation}
As a result, the PIN estimator transfer matrix $C_{\text{PIN}}(s,\omega_\mathrm{r})$ is formulated as:
\begin{equation}
    \left[
    \begin{matrix}
       \hat{U}_1(s) \\
        \hat{U}_2(s) \\
        \hat{U}_3(s)
    \end{matrix}
    \right]
:=
\underbrace{\left[
    \begin{matrix}
       K(s) &  0 & 0 \\
       0 & K(s) & 0  \\
       0 & 0 & K(s)
    \end{matrix}
    \right]}_{C_{\text{PIN}}(s,\omega_\mathrm{r})}
\left[
    \begin{matrix}
       \epsilon_1(s)  \\
       \epsilon_2(s)  \\
       \epsilon_3(s)
    \end{matrix}
    \right].
 \label{PIN estimator}
\end{equation}



\begin{figure}
\centering 
\includegraphics[width=1\columnwidth]{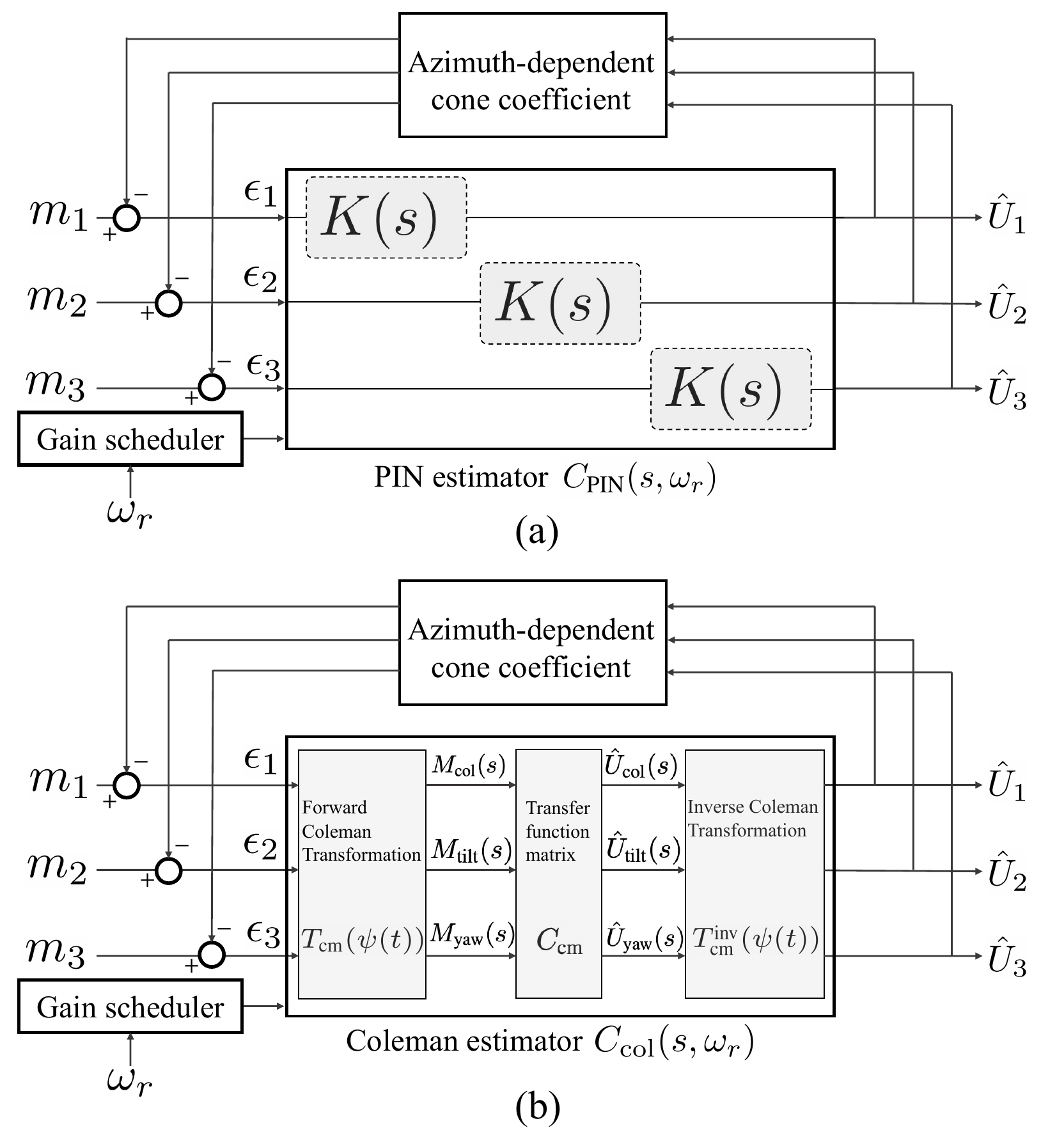}
\caption{Block diagram of the BEWS estimators formed by a negative feedback interconnection of a linear estimator gain and a nonlinear azimuth-dependent cone coefficient. (a) PIN estimator, (b) Coleman estimator.}
\label{block_diagram} %
\end{figure}

From Eq.~\eqref{PIN estimator}, it is evident that a diagonal and decoupled BEWS estimator $K(s)$ solely responds to its local blade load measurements which is depicted in Fig.~\ref{block_diagram}(a).
However, it is well known that there exists a significant amount of dynamic coupling between the blades of a turbine rotor~\cite{Mulders_2019}, that is, the coupling is present between the wind speed $\hat{U}_i(s)$ and the blade load $m_j(s)$ for $i \ne j$.
The PIN estimator $C_{\text{PIN}}(s,\omega_\mathrm{r})$ is diagonal and thus the off-diagonal terms are nonexistent, which implies that the proposed estimator is unable to take the coupling effects between blades into account.

The lack of information on the coupled blade dynamics will degrade the performance of the wind speed estimator, as will become apparent during the numerical simulations in Section~\ref{sec:4}. 
To cope with the problem of unmodeled dynamics, we propose a second estimation technique which employs the Coleman transformation in a similar estimation framework. This will be outlined in the next subsection.
Afterwards, the resemblances between and differences of both estimators are discussed.

\subsection{Coleman estimator}
\noindent The Coleman transformation~\cite{Bir_2008} is an important technique to convert the wind turbine dynamics to the non-rotating frame. 
By replacing the diagonal part of the PIN estimator with the forward and inverse Coleman transformations and by including a diagonal integral transfer matrix, the so-called Coleman estimator is formulated as shown in Fig.~\ref{block_diagram}(b).

The forward Coleman transformation used in the estimator is defined as follows:
\begin{multline}
\left[
    \begin{matrix}
        \epsilon_{\text{col}} \\
        \epsilon_{\text{tilt}} \\
        \epsilon_{\text{yaw}}
    \end{matrix}
    \right] 
    :=
    \underbrace{
    \frac{2}{3} 
    \left [
    \begin{matrix}
    {1}/{2} & {1}/{2} & {1}/{2} \\
    \sin{(\psi_1)} & \sin{\left(\psi_2\right)} & \sin{\left(\psi_3\right)} \\
    \cos{(\psi_1)} & \cos{\left(\psi_2\right)} & \cos{\left(\psi_3\right)}
    \end{matrix}
    \right ]}_{T_{\text{cm}}(\psi(t))}
    \left[
    \begin{matrix}
        \epsilon_{1} \\
        \epsilon_{2} \\
        \epsilon_{3}
    \end{matrix}
    \right],
    \label{Coleman transformation}
\end{multline}
where $\epsilon_{\text{tilt}}$ and $\epsilon_{\text{yaw}}$ are the error signals of the estimated and projected MOoP in tilt and yaw directions, respectively, while $\epsilon_{\text{col}}$ indicates that of the collective MOoP estimation.

With the transfer function matrix $C_{\text{cm}}$ these quantities are mapped to the estimated collective, tilt and yaw wind speed components. The matrix $C_{\text{cm}}$ is composed out of pure integrals in a diagonal structure to obtain the estimate of the wind speeds in the non-rotating frame, and is defined as:
\begin{equation}
    \left[
    \begin{matrix}
        \hat{U}_{\text{col}} \\
        \hat{U}_{\text{tilt}} \\
        \hat{U}_{\text{yaw}}
    \end{matrix}
    \right] 
    :=
    \underbrace{
    \left [
    \begin{matrix}
    {K_{\text{col}}}/{s} & 0 & 0 \\
    0 & {K_{0}}/{s} & 0 \\
    0 & 0 & {K_{0}}/{s}
    \end{matrix}
    \right ]}_{C_{\text{cm}}(s)} 
    \left[
    \begin{matrix}
        \epsilon_{\text{col}} \\
        \epsilon_{\text{tilt}} \\
        \epsilon_{\text{yaw}}
    \end{matrix}
    \right],
    \label{Ccm}
\end{equation}
where $K_{\text{col}}$ is the integral gain for estimating the collective component of the wind speed whereas the gain $K_0$ is used for both the tilt and yaw referred wind speeds for simplicity. Finally, the signals are projected back into the rotating frame to obtain BEWS with the inverse Coleman transformation
\begin{align}
\left[
    \begin{matrix}
        \hat{U}_{1} \\
        \hat{U}_{2} \\
        \hat{U}_{3}
    \end{matrix}
    \right] 
    &:=
    \underbrace{
    \left [
    \begin{matrix}
    1 & \sin{(\psi_1)} & \cos{(\psi_1)} \\
    1 & \sin{\left(\psi_2\right)} & \cos{\left(\psi_2\right)} \\
    1 & \sin{\left(\psi_3\right)} & \cos{\left(\psi_3\right)} 
    \end{matrix}
    \right ]}_{T_{\text{cm}}^{\text{inv}}(\psi(t))}
    \left[
    \begin{matrix}
        \hat{U}_{\text{col}} \\
        \hat{U}_{\text{tilt}} \\
        \hat{U}_{\text{yaw}}
    \end{matrix}
    \right] 
    \,.
    \label{inverse Coleman}
\end{align}



\subsection{Similarities between PIN and Coleman estimators}
\noindent Given these two seeming dissimilar BEWS estimators, it is compelling to attempt to understand their similarities. 
To achieve this goal, the transfer function matrix form of the Coleman estimator gain is first derived in Proposition~\ref{prop: transfer function matrix}.

\begin{prop}
\label{prop: transfer function matrix}
Let us consider the wind turbine operating at a constant rotor speed, that is, $\omega_\mathrm{r}(t)=\omega_0$.
The Coleman estimator gain $C_\mathrm{col}$ consisting of ~\eqref{Coleman transformation}, \eqref{Ccm} and \eqref{inverse Coleman} is equivalent to the following transfer function matrix:

\begin{equation}
C_\mathrm{col}(s,\omega_0) := 
\left[
    \begin{matrix}
       K_\mathrm{R,a}(s) &K_\mathrm{R,b}(s) &K_\mathrm{R,c}(s) \\
       K_\mathrm{R,c}(s) &K_\mathrm{R,a}(s) &K_\mathrm{R,b}(s) \\
       K_\mathrm{R,b}(s) &K_\mathrm{R,c}(s) &K_\mathrm{R,a}(s)
    \end{matrix}
    \right] 
    \,,
    \label{eq:Coleman transfer function matrix}
\end{equation}
where the subscript '$R$' the rotating reference frame, and the transfer functions $K_\mathrm{R,a}$, $K_\mathrm{R,b}$ and $K_\mathrm{R,c}$ are defined as:
\begin{align}
   K_\mathrm{R,a}(s) &:= \frac{(2 K_0 + K_\mathrm{col})s^2+ K_\mathrm{col}\omega_0^2}{3 s(s^2+\omega_0^2)}
   ,\nonumber\\
   K_\mathrm{R,b}(s) &:= \frac{(K_\mathrm{col}-K_0)s^2 + ( K_0 \sqrt{3}\omega_0 )s + K_\mathrm{col}\omega_0^2}{3 s(s^2+\omega_0^2)}
   ,\nonumber\\
   K_\mathrm{R,c}(s) &:= \frac{(K_\mathrm{col}-K_0)s^2 - ( K_0 \sqrt{3}\omega_0 )s + K_\mathrm{col}\omega_0^2}{3 s(s^2+\omega_0^2)}\nonumber.
\end{align}
\end{prop}

\begin{proof}
The transfer function between the estimated wind speed in the non-rotating frame $\hat{U}_\text{col,tilt,yaw}(s)$ and $\hat{U}_{1,2,3}(s)$ in the rotating reference frame is given by
\begin{multline}
    \left[
    \begin{matrix}
        \hat{U}_1(s) \\
        \hat{U}_2(s) \\
        \hat{U}_3(s) 
    \end{matrix}
    \right]
    =
    \mathcal{C}_{-}^\mathrm{T}
    \left[
    \begin{matrix}
        \hat{U}_\text{col}(s-\mathrm{j}\omega_\mathrm{r}) \\
        \hat{U}_\text{tilt}(s-\mathrm{j}\omega_\mathrm{r}) \\
        \hat{U}_\text{yaw}(s-\mathrm{j}\omega_\mathrm{r}) 
    \end{matrix}
    \right]
    +
    \mathcal{C}_{+}^\mathrm{T}
    \\
    \left[
    \begin{matrix}
        \hat{U}_\text{col}(s+\mathrm{j}\omega_\mathrm{r}) \\
        \hat{U}_\text{tilt}(s+\mathrm{j}\omega_\mathrm{r}) \\
        \hat{U}_\text{yaw}(s+\mathrm{j}\omega_\mathrm{r}) 
    \end{matrix}
    \right]
    +
    \mathcal{C}_\text{col}^\mathrm{T}
    \left[
    \begin{matrix}
        \hat{U}_\text{col}(s) \\
        \hat{U}_\text{tilt}(s) \\
        \hat{U}_\text{yaw}(s) 
    \end{matrix}
    \right],
\label{q:thetaTY_to_theta123}
\end{multline}
where
\begin{multline*}
    \mathcal{C}_{-} = 
    \frac{1}{2}
    \left[
    \begin{matrix}
        0 & 0  & 0 \\
        0 & 1  & \mathrm{j} \\
        0 & -\mathrm{j} & 1
    \end{matrix}
    \right]
    \left[
    \begin{matrix}
        0 & 0 & 0 \\
        \cos{(0)} & \cos{(\frac{2\pi}{3})} & \cos{(\frac{4\pi}{3})} \\
        \sin{(0)} & \sin{(\frac{2\pi}{3})} & \sin{(\frac{4\pi}{3})}
    \end{matrix}
    \right],
\end{multline*}
\begin{multline*}
    \mathcal{C}_{+} = 
    \frac{1}{2} 
    \left[
    \begin{matrix}
        0 & 0  & 0 \\
        0 & 1  & -\mathrm{j} \\
        0 & \mathrm{j} & 1
    \end{matrix}
    \right]
    \left[
    \begin{matrix}
        0 & 0 & 0 \\
        \cos{(0)} & \cos{(\frac{2\pi}{3})} & \cos{(\frac{4\pi}{3})} \\
        \sin{(0)} & \sin{(\frac{2\pi}{3})} & \sin{(\frac{4\pi}{3})}
    \end{matrix}
    \right],
\end{multline*}
and
\begin{equation*}
    \mathcal{C}_\text{col} = 
    \left[
    \begin{matrix}
        1 & 1  & 1 \\
        0 & 0  & 0 \\
        0 & 0  & 0
    \end{matrix}
    \right],
\end{equation*}
with $\mathrm{j}=\sqrt{-1}$ being the imaginary unit. Given the diagonal integral form of Eq.~\eqref{Ccm}, it is straightforward to derive the frequency-shifted transfer function as:
\begin{equation}
\label{eq: tf1}
    \left[
    \begin{matrix}
        \hat{U}_\text{col}(s-\mathrm{j}\omega_\mathrm{r})  \\
        \hat{U}_\text{tilt}(s-\mathrm{j}\omega_\mathrm{r}) \\
        \hat{U}_\text{yaw}(s-\mathrm{j}\omega_\mathrm{r}) 
    \end{matrix}
    \right]
    =
    C_\text{cm}(s-\mathrm{j}\omega_\mathrm{r})
    \left[
    \begin{matrix}
        \epsilon_\text{col}(s-\mathrm{j}\omega_\mathrm{r})  \\
        \epsilon_\text{tilt}(s-\mathrm{j}\omega_\mathrm{r}) \\
        \epsilon_\text{yaw}(s-\mathrm{j}\omega_\mathrm{r}) 
    \end{matrix}
    \right],
\end{equation}
and 
\begin{equation}
\label{eq: tf2}
    \left[
    \begin{matrix}
        \hat{U}_\text{col}(s+\mathrm{j}\omega_\mathrm{r})  \\
        \hat{U}_\text{tilt}(s+\mathrm{j}\omega_\mathrm{r}) \\
        \hat{U}_\text{yaw}(s+\mathrm{j}\omega_\mathrm{r}) 
    \end{matrix}
    \right]
    =
    C_\text{cm}(s+\mathrm{j}\omega_\mathrm{r})
    \left[
    \begin{matrix}
        \epsilon_\text{col}(s+\mathrm{j}\omega_\mathrm{r})  \\
        \epsilon_\text{tilt}(s+\mathrm{j}\omega_\mathrm{r}) \\
        \epsilon_\text{yaw}(s+\mathrm{j}\omega_\mathrm{r}) 
    \end{matrix}
    \right].
\end{equation}

\noindent Combining Eqs.~\eqref{eq: tf1} and~\eqref{eq: tf2}, Eq.~\eqref{q:thetaTY_to_theta123} is rewritten as:
\begin{multline}
\label{eq:MTY_to_theta123}
    \left[
    \begin{matrix}
        \hat{U}_1(s) \\
        \hat{U}_2(s) \\
        \hat{U}_3(s) 
    \end{matrix}
    \right]
    =
    \mathcal{C}_{-}^\mathrm{T}
    C_\text{cm}(s-\mathrm{j}\omega_\mathrm{r})
    \left[
    \begin{matrix}
        \epsilon_\text{col}(s-\mathrm{j}\omega_\mathrm{r}) \\
        \epsilon_\text{tilt}(s-\mathrm{j}\omega_\mathrm{r}) \\
        \epsilon_\text{yaw}(s-\mathrm{j}\omega_\mathrm{r}) 
    \end{matrix}
    \right]\\
    +
    \mathcal{C}_{+}^\mathrm{T}
    C_\text{cm}(s+\mathrm{j}\omega_\mathrm{r})
    \left[
    \begin{matrix}
        \epsilon_\text{col}(s+\mathrm{j}\omega_\mathrm{r}) \\
        \epsilon_\text{tilt}(s+\mathrm{j}\omega_\mathrm{r}) \\
        \epsilon_\text{yaw}(s+\mathrm{j}\omega_\mathrm{r}) 
    \end{matrix}
    \right]
    +
    \mathcal{C}_\text{col}^\mathrm{T}
    C_\text{cm}(s)
    \left[
    \begin{matrix}
        \epsilon_\text{col}(s) \\
        \epsilon_\text{tilt}(s) \\
        \epsilon_\text{yaw}(s) 
    \end{matrix}
    \right].
\end{multline}

\noindent In addition, the Coleman transformed expressions of $\epsilon_{1,2,3}$ are formulated in the frequency domain as:
\begin{multline}
\label{eq:MTY_to_M123}
    \left[
    \begin{matrix}
        \epsilon_\text{col}(s) \\
        \epsilon_\text{tilt}(s) \\
        \epsilon_\text{yaw}(s)
    \end{matrix}
    \right]
    =
    \frac{2}{3}
    \mathcal{C}_{-}
    \left[
    \begin{matrix}
        \epsilon_1(s-\mathrm{j}\omega_\mathrm{r}) \\
        \epsilon_2(s-\mathrm{j}\omega_\mathrm{r}) \\
        \epsilon_3(s-\mathrm{j}\omega_\mathrm{r})
    \end{matrix}
    \right]
    +
    \\
    \frac{2}{3}
    \mathcal{C}_{+}
    \left[
    \begin{matrix}
        \epsilon_1(s+\mathrm{j}\omega_\mathrm{r}) \\
        \epsilon_2(s+\mathrm{j}\omega_\mathrm{r}) \\
        \epsilon_3(s+\mathrm{j}\omega_\mathrm{r})
    \end{matrix}
    \right]
    +
    \frac{1}{3}
    \mathcal{C}_\text{col}
    \left[
    \begin{matrix}
        \epsilon_1(s) \\
        \epsilon_2(s) \\
        \epsilon_3(s)
    \end{matrix}
    \right],
\end{multline}

\noindent Substituting Eq.~\eqref{eq:MTY_to_M123} into Eq.~\eqref{eq:MTY_to_theta123} yields
\begin{multline}
\label{eq:MTY_to_theta123_simple}
    \left[
    \begin{matrix}
        \hat{U}_1(s) \\
        \hat{U}_2(s) \\
        \hat{U}_3(s) 
    \end{matrix}
    \right]
    =
    \mathcal{C}_{-}^\mathrm{T}
    C_\text{cm}(s-\mathrm{j}\omega_\mathrm{r})
    \left(
        \frac{2}{3}
        \mathcal{C}_{-}
        \left[
        \begin{matrix}
            \epsilon_1(s-2\mathrm{j}\omega_\mathrm{r}) \\
            \epsilon_2(s-2\mathrm{j}\omega_\mathrm{r}) \\
            \epsilon_3(s-2\mathrm{j}\omega_\mathrm{r})
        \end{matrix}
        \right]\right.
        \\
        +
        \left.
        \frac{2}{3}
        \mathcal{C}_{+}
        \left[
        \begin{matrix}
            \epsilon_1(s) \\
            \epsilon_2(s) \\
            \epsilon_3(s)
        \end{matrix}
        \right]
        +
        \frac{1}{3}
        \mathcal{C}_\text{col}
        \left[
        \begin{matrix}
            \epsilon_1(s-\mathrm{j}\omega_\mathrm{r}) \\
            \epsilon_2(s-\mathrm{j}\omega_\mathrm{r}) \\
            \epsilon_3(s-\mathrm{j}\omega_\mathrm{r})
        \end{matrix}
        \right]
    \right)
    \\
    +
    \mathcal{C}_{+}^\mathrm{T}
    C_\text{cm}(s+\mathrm{j}\omega_\mathrm{r})
    \left(
        \frac{2}{3}
        \mathcal{C}_{-}
        \left[
        \begin{matrix}
            \epsilon_1(s) \\
            \epsilon_2(s) \\
            \epsilon_3(s)
        \end{matrix}
        \right]
        +
        \frac{2}{3}
        \mathcal{C}_{+}
        \left[
        \begin{matrix}
            \epsilon_1(s+2\mathrm{j}\omega_\mathrm{r}) \\
            \epsilon_2(s+2\mathrm{j}\omega_\mathrm{r}) \\
            \epsilon_3(s+2\mathrm{j}\omega_\mathrm{r})
        \end{matrix}
        \right]\right.
        \\
        +
        \left.
        \frac{1}{3}
        \mathcal{C}_\text{col}
        \left[
        \begin{matrix}
            \epsilon_1(s+\mathrm{j}\omega_\mathrm{r}) \\
            \epsilon_2(s+\mathrm{j}\omega_\mathrm{r}) \\
            \epsilon_3(s+\mathrm{j}\omega_\mathrm{r})
        \end{matrix}
        \right]
    \right)
    +
    \mathcal{C}_\text{col}^\mathrm{T}
    C_\text{cm}(s)
    \left(
        \frac{2}{3}
        \mathcal{C}_{-}
        \left[
        \begin{matrix}
            \hat{U}_1(s-\mathrm{j}\omega_\mathrm{r}) \\
            \hat{U}_2(s-\mathrm{j}\omega_\mathrm{r}) \\
            \hat{U}_3(s-\mathrm{j}\omega_\mathrm{r})
        \end{matrix}
        \right]\right.
        \\
        +
        \left.
        \frac{2}{3}
        \mathcal{C}_{+}
        \left[
        \begin{matrix}
            \epsilon_1(s+\mathrm{j}\omega_\mathrm{r}) \\
            \epsilon_2(s+\mathrm{j}\omega_\mathrm{r}) \\
            \epsilon_3(s+\mathrm{j}\omega_\mathrm{r})
        \end{matrix}
        \right]
        +
        \frac{1}{3}
        \mathcal{C}_\text{col}
        \left[
        \begin{matrix}
            \epsilon_1(s) \\
            \epsilon_2(s) \\
            \epsilon_3(s)
        \end{matrix}
        \right]
    \right),
    \end{multline}
which can be simplified by the remarkable property: 
\begin{align}
\mathcal{C}_{-}^\mathrm{T}\mathcal{C}_{-} = \mathcal{C}_{+}^\mathrm{T}\mathcal{C}_{+} = \mathcal{C}_{-}^\mathrm{T}\mathcal{C}_\text{col} = \mathcal{C}_{+}^\mathrm{T}\mathcal{C}_\text{col} = \mathcal{C}_\text{col}^\mathrm{T}\mathcal{C}_{-} = \mathcal{C}_\text{col}^\mathrm{T}\mathcal{C}_{+} = 0,\nonumber
\end{align}
and results in the following equation:
\begin{multline}
\label{eq:MTY_to_theta123_simplified}
    \left[
    \begin{matrix}
        \hat{U}_1(s) \\
        \hat{U}_2(s) \\
        \hat{U}_3(s) 
    \end{matrix}
    \right]
    =
    \left(
    \frac{2}{3}
    \mathcal{C}_{-}^\mathrm{T}
    C_\text{cm}(s-\mathrm{j}\omega_\mathrm{r})
    \mathcal{C}_{+}
    \right.
    \\
    \left.
    +
    \frac{2}{3}
    \mathcal{C}_{+}^\mathrm{T}
    C_\text{cm}(s+\mathrm{j}\omega_\mathrm{r})
    \mathcal{C}_{-}
    +
    \frac{1}{3}
    \mathcal{C}_\text{col}^\mathrm{T}
    C_\text{cm}(s)
    \mathcal{C}_\text{col}
    \right)
    \left[
    \begin{matrix}
        \epsilon_1(s) \\
        \epsilon_2(s) \\
        \epsilon_3(s)
    \end{matrix}
    \right],
\end{multline}
which results in the transfer function matrix form of the proposed Coleman estimator 
\label{eq:MBCestimator}
and completes the proof.
\end{proof}
Bode magnitude plots of the transfer function matrix $C_{\text{col}}(s,\omega_0)$ are illustrated in Fig.~\ref{MBCgain}, where $\omega_\mathrm{r}=2\pi f_r$ with $f_r = 0.2$~[Hz] refers to a constant $1$P frequency of the considered wind turbine model.
\begin{figure}
\centering 
\includegraphics[width=1\columnwidth]{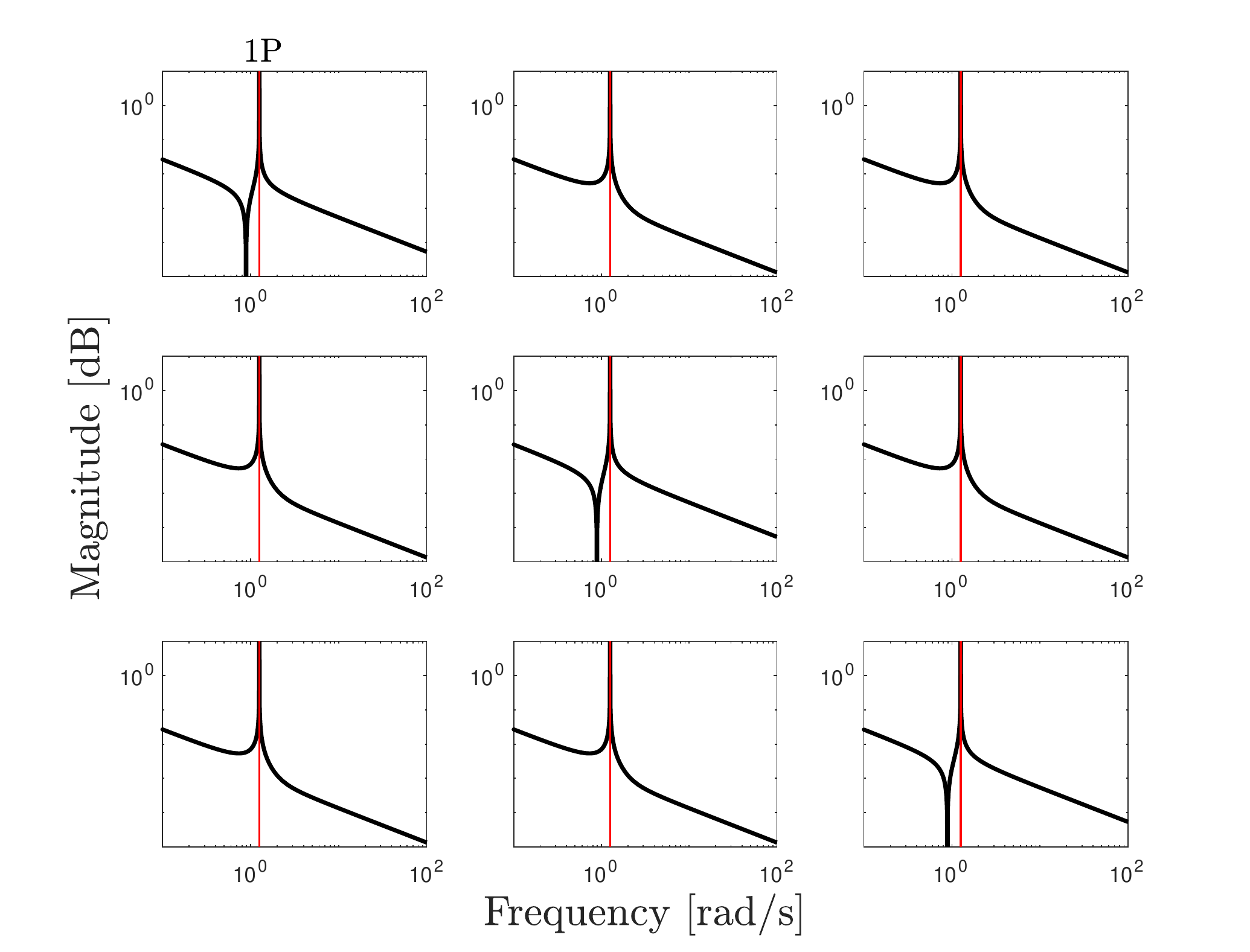}
\caption{Bode magnitude plots of the Coleman estimator defined by the transfer function matrix in Eq.~\eqref{eq:Coleman transfer function matrix}, where 1P frequency is assumed as constant, which is indicated by a red vertical line.}
\label{MBCgain} %
\end{figure}


The similarities between the PIN estimator (Eq.~\eqref{PIN estimator}) and the Coleman estimator (Eq.~\eqref{eq:Coleman transfer function matrix}) are now further illustrated in the frequency domain.
\begin{theorem}
Consider the wind turbine operating at a constant rotor speed $\omega_0$.
Given the nonlinear and azimuth-dependent cone coefficient of Eq.~\eqref{Cm table}, $C_\mathrm{PIN}(s,\omega_\mathrm{r})$ defined in Eq.~\eqref{PIN estimator} is equivalent to a diagonal structure of $C_\mathrm{col}(s)$ in Eq.~\eqref{eq:Coleman transfer function matrix} only including $K_\mathrm{R,a}(s)$, where
\begin{align}
k_\mathrm{i} &= \frac{K_\mathrm{col}}{3},\qquad
k_\mathrm{p} = \frac{K_0}{3\omega_0}.
\label{condition 1}
\end{align}
\label{theorm 1}
\end{theorem}
\begin{proof}
To prove Theorem~\ref{theorm 1}, the equality $K(s) = K_\mathrm{R,a}(s)$ needs to hold under the conditions of Eq.~\eqref{condition 1}.
It is known from Eqs.~\eqref{PIN equation} and \eqref{eq:Notch-peak} that 
\begin{align}
K(s) &= k_\mathrm{p} \frac{2\omega_\mathrm{r} s}{s^2+\omega_\mathrm{r}^2} + \frac{k_\mathrm{i}}{s} \\
&= \frac{k_\mathrm{p} s}{s} \frac{2\omega_\mathrm{r} s}{s^2+\omega_\mathrm{r}^2}+\frac{k_\mathrm{i}}{s} \frac{s^2+\omega_\mathrm{r}^2}{s^2+\omega_\mathrm{r}^2} \\
&= \frac{(2k_\mathrm{p}\omega_\mathrm{r} +k_\mathrm{i})s^2+k_\mathrm{i} \omega_\mathrm{r}^2}{s(s^2+\omega_\mathrm{r}^2)}.
\label{proof of thereom 1}
\end{align}
Given condition~\eqref{condition 1}, it is evident that \eqref{proof of thereom 1} is equal to $K_\mathrm{R,a}$ in~\eqref{eq:Coleman transfer function matrix} which proves that the PIN estimator is equivalent to a particular diagonal structure of the Coleman estimator.
\end{proof}

\section{Case study}\label{sec:4}
\noindent This section performs a case study to verify the theoretical analysis and to demonstrate the similarities between PIN and Coleman estimators.

NREL's high-fidelity wind turbine simulation software package Fatigue, Aerodynamics, Structures, and Turbulence (FAST)~\cite{Jonkman-2005} is utilized to simulate the wind turbine dynamics.
The closed-loop control system including both wind speed estimator types is implemented in Simulink. The baseline \textit{K-omega-squared} torque control law~\cite{Jonkman-2005}, with a predefined optimal mode gain, is used in this letter. To evaluate the performance of the estimators, a sheared wind profile with a step-wise increasing wind speed from $8$~m/s to $10$~m/s is considered. The gains of the wind speed estimators are determined according to Theorem~\ref{theorm 1}, and the simulation time step is set to $0.01$~s.

\begin{figure}
\centering 
\includegraphics[width=1\columnwidth]{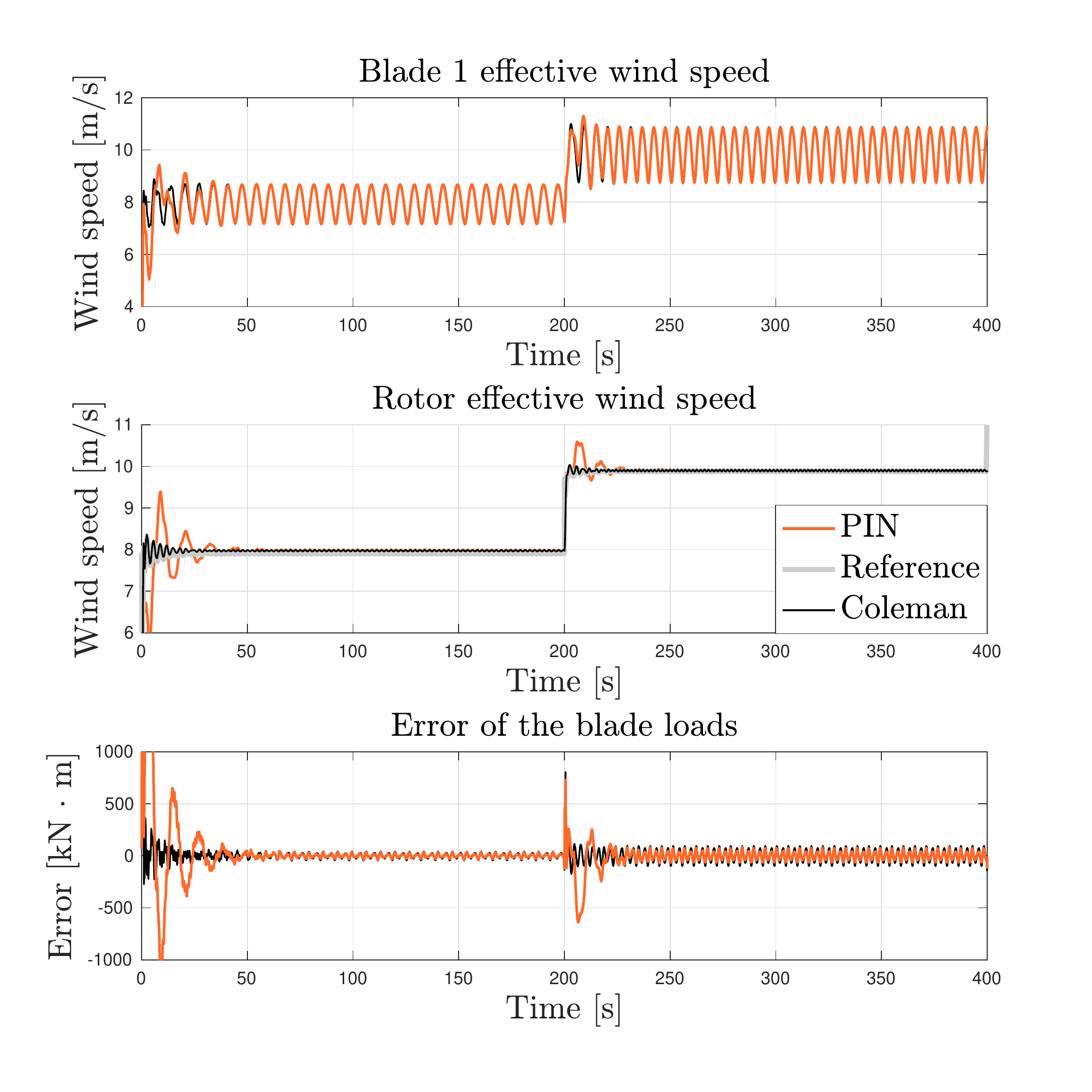}
\caption{Time history of the wind speed and the estimation error on blade 1 between the PIN and  Coleman estimators. Blades 2-3 show similar results and hence omitted.}
\label{Model compare} %
\end{figure}

The BEWS estimation results of the PIN and Coleman estimators are presented in Fig.~\ref{Model compare}. The proposed estimators show similar results for the considered step-wise sheared wind conditions, which substantiates Theorem~\ref{theorm 1}. The power spectra of the estimated wind speed calculated for $200$\,s\,--\,$400$\,s are presented in Fig.~\ref{power spectrum of wind}. The spectrum of the actual wind speed the blade experiences is also calculated and included as a reference and benchmark to evaluate estimator performance. 

The Coleman and PIN estimators show almost consistent performance for the dominant $1$P harmonic. For the other frequencies, the power spectrum of the Coleman estimator is closer to the reference value. This indicates that, compared to the PIN estimator, the Coleman estimator has a superior estimation quality for non-$1$P frequencies. In addition, it is evident that the Coleman estimator shows less oscillations during transients than the PIN estimator, which leads to a smoother and faster tracking of the actual wind speed.
Such a smoother transient response is explained by Eq.~\eqref{eq:Coleman transfer function matrix}, as the MIMO transmission zeros are cancelled in this full transfer function matrix form. 

In summary, both wind speed estimators show similar performance for the considered wind conditions, which verifies Theorem~\ref{theorm 1}. The PIN estimator, which is deemed as a particular diagonal form of the Coleman estimator, is a simple but effective method to estimate the BEWS. The Coleman estimator, on the other hand, is superior in considering the coupled rotor dynamics for the BEWS estimation. Since the effects of the transmission zeros are eliminated by its off-diagonal terms, the Coleman estimator exhibits a smoother transient response than the PIN estimator.

\begin{figure}
\centering 
\includegraphics[width=1\columnwidth]{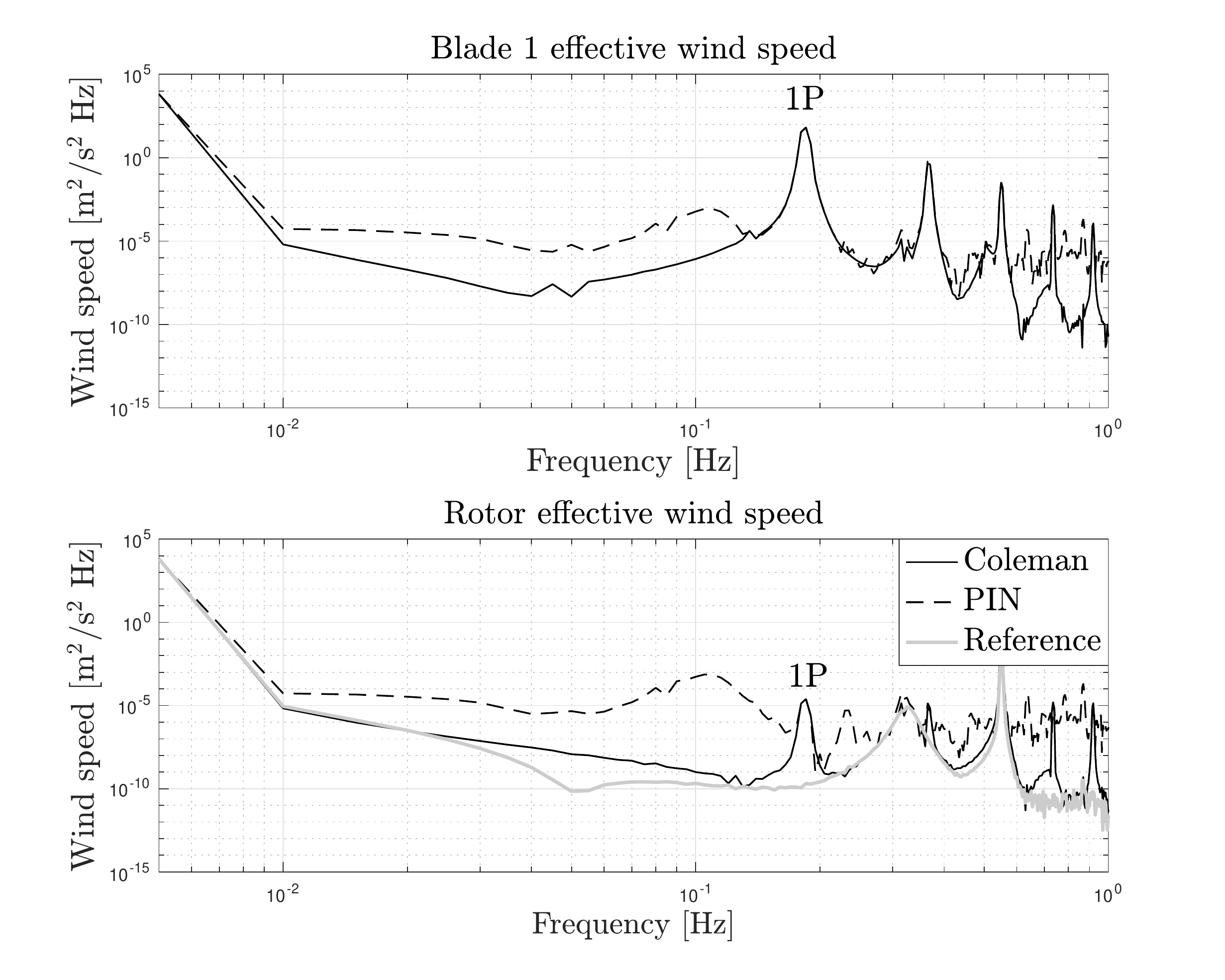}
\caption{Simulation results on the NREL 5MW wind turbine, where the power spectrum of the wind speed on blade 1 illustrates the performance similarities between the PIN and Coleman estimators. Blades 2-3 show similar results and hence omitted.}
\label{power spectrum of wind} %
\end{figure}

%
%

\section{Conclusions}\label{sec:5}
\noindent Two novel Blade Effective Wind Speed (BEWS) estimation schemes are proposed and an analysis is performed on the similarities of both implementations. First, a Proportional-Integral-Notch (PIN) estimator is introduced, which is formed by a negative feedback interconnection of a linear diagonally decoupled estimator and a nonlinear azimuth-dependent cone coefficient. Following the same philosophy, a Coleman estimator, which is based on the Coleman transformation, is developed to incorporate the coupled rotor dynamics into the estimation scheme.

Given the seeming disparities between these two estimators, it is proven that there exists structural and performance similarities between the PIN and Coleman estimators. Analytical and numerical results show the diagonal PIN and coupled Coleman estimators exhibit similar estimation results. In particular, the PIN estimator, which is considered as a diagonal form of the Coleman estimator, is a simple but effective method to estimate the BEWS. The Coleman estimator, on the other hand, is superior in taking into account the coupled rotor dynamics for the BEWS estimation. Since the effects of the MIMO transmission zeros present in the PIN scheme are eliminated by the off-diagonal terms of the Coleman estimator, the latter mentioned shows better transient behavior compared to the PIN estimator.

Future work will include the investigation of global convergence properties of the wind speed estimator, coupled with wind turbine controllers such as a tip-speed ratio tracker.


\bibliographystyle{IEEEtran} 
\bibliography{references}

\begin{thebibliography}{10}
\providecommand{\url}[1]{#1}
\csname url@samestyle\endcsname
\providecommand{\newblock}{\relax}
\providecommand{\bibinfo}[2]{#2}
\providecommand{\BIBentrySTDinterwordspacing}{\spaceskip=0pt\relax}
\providecommand{\BIBentryALTinterwordstretchfactor}{4}
\providecommand{\BIBentryALTinterwordspacing}{\spaceskip=\fontdimen2\font plus
\BIBentryALTinterwordstretchfactor\fontdimen3\font minus
  \fontdimen4\font\relax}
\providecommand{\BIBforeignlanguage}[2]{{%
\expandafter\ifx\csname l@#1\endcsname\relax
\typeout{** WARNING: IEEEtran.bst: No hyphenation pattern has been}%
\typeout{** loaded for the language `#1'. Using the pattern for}%
\typeout{** the default language instead.}%
\else
\language=\csname l@#1\endcsname
\fi
#2}}
\providecommand{\BIBdecl}{\relax}
\BIBdecl

\bibitem{GWEC_2021}
{Global Wind Energy Council}, ``Global wind report 2021,'' Global wind energy
  council, Report, 2021.

\bibitem{Soltani_2013}
M.~N. {Soltani}, T.~{Knudsen}, M.~{Svenstrup}, R.~{Wisniewski}, P.~{Brath},
  R.~{Ortega}, and K.~{Johnson}, ``Estimation of rotor effective wind speed: A
  comparison,'' \emph{IEEE Transactions on Control Systems Technology},
  vol.~21, no.~4, pp. 1155--1167, July 2013.

\bibitem{Ortega_2013}
R.~Ortega, F.~Mancilla-David, and F.~Jaramillo, ``A globally convergent wind
  speed estimator for wind turbine systems,'' \emph{International Journal of
  Adaptive Control and Signal Processing}, vol.~27, no.~5, pp. 413--425, 2013.

\bibitem{Liu_LCSS2021}
Y.~Liu, A.~K. Pamososuryo, R.~M.~G. Ferrari, and J.~W. van Wingerden, ``The
  immersion and invariance wind speed estimator revisited and new results,''
  \emph{IEEE Control Systems Letters}, vol.~6, pp. 361--366, 2022.

\bibitem{Bossanyi_2000}
E.~A. Bossanyi, ``The design of closed loop controllers for wind turbines,''
  \emph{Wind Energy}, vol.~3, no.~3, pp. 149--163, 2000.

\bibitem{Ostergaard2007}
K.~Z. {\O}stergaard, P.~Brath, and J.~Stoustrup, ``Estimation of effective wind
  speed,'' in \emph{Journal of Physics: Conference Series}, vol.~75,
  no.~1.\hskip 1em plus 0.5em minus 0.4em\relax The Science of Making Torque
  from Wind Lyngby, Denmark, 2007, p. 012082.

\bibitem{Bottasso_2018}
C.~Bottasso, S.~Cacciola, and J.~Schreiber, ``Local wind speed estimation, with
  application to wake impingement detection,'' \emph{Renewable Energy}, vol.
  116, pp. 155 -- 168, 2018.

\bibitem{Liu2021}
Y.~{Liu}, A.~{Kusumo Pamososuryo}, R.~{Ferrari}, T.~{Gybel Hovgaard}, and J.~W.
  {van Wingerden}, ``{Blade Effective Wind Speed Estimation: A Subspace
  Predictive Repetitive Estimator Approach},'' in \emph{2021 European Control
  Conference (ECC)}, 2021.

\bibitem{Bossanyi2003}
\BIBentryALTinterwordspacing
E.~A. Bossanyi, ``{Individual Blade Pitch Control for Load Reduction},''
  \emph{Wind Energy}, vol.~6, no.~2, pp. 119--128, apr 2003. [Online].
  Available: \url{http://doi.wiley.com/10.1002/we.76}
\BIBentrySTDinterwordspacing

\bibitem{Mulders_2019}
S.~P. Mulders, A.~K. Pamososuryo, G.~E. Disario, and J.~W. {van Wingerden},
  ``Analysis and optimal individual pitch control decoupling by inclusion of an
  azimuth offset in the multiblade coordinate transformation,'' \emph{Wind
  Energy}, vol.~22, no.~3, pp. 341--359, 2019.

\bibitem{Jonkman_2009}
J.~Jonkman, S.~Butterfield, W.~Musial, and G.~Scott, ``Definition of a 5{MW}
  reference wind turbine for offshore system development,'' National Renewable
  Energy Laboratory, NREL/TP-500-38060, 2009.

\bibitem{Bir_2008}
G.~Bir, ``Multi-blade coordinate transformation and its application to wind
  turbine analysis,'' in \emph{46th AIAA Aerospace Sciences Meeting and
  Exhibit}, 2008.

\bibitem{Jonkman-2005}
J.~M. Jonkman and M.~L. Buhl, ``Fast user's guide,'' \emph{SciTech Connect:
  FAST User's Guide}, 2005.

\end{thebibliography}


\end{document}